\documentclass[10pt]{article}
\usepackage{amsmath,amssymb,amsfonts,amsthm}
\usepackage{fullpage}
\usepackage{hyperref,doi}
\usepackage{lmodern}
\usepackage{dsfont}
\usepackage[numbers]{natbib}

\theoremstyle{definition}

\newtheorem{theorem}{Theorem}
\newtheorem{lemma}{Lemma}

\newtheorem{remark}{Remark}

\newcommand{\R}{\mathbb{R}}
\newcommand{\C}{\mathbb{C}}

\newcommand{\N}{\mathbb{N}}

\newcommand{\id}{\mathds{1}}
\newcommand{\eps}{\varepsilon}

\newcommand{\Z}{\mathbb{Z}}
\newcommand{\e}{\mathrm{e}}
\newcommand{\sgn}{\mathop{\mathrm{sgn}}}
\newcommand{\abs}[1]{| #1 |}
\newcommand{\babs}[1]{\left| #1 \right|}
\newcommand{\crangle}{\rangle_{H}}
\newcommand{\srangle}{\rangle_{S}}
\newcommand{\norm}[1]{\left\| #1 \right\| }
\newcommand{\opnorm}[1]{\left\| #1 \right\|_{\mathrm{op}} }
\newcommand{\cnorm}[1]{\left\| #1 \right\|_{H}}
\DeclareMathOperator*{\Exp}{\mathbb{E}}
\DeclareMathOperator*{\tr}{\mathrm{tr}}
\newcommand{\dz}{\mathrm{d}z}
\newcommand{\ddz}{\frac{\mathrm{d}}{\dz}}
\newcommand{\dddz}{\frac{\mathrm{d}^2}{\dz^2}}
\newcommand{\dx}{\mathrm{d}x}

\newcommand{\dr}{\mathrm{d}r}

\newcommand{\hatf}{\widehat{f}}
\newcommand{\hatg}{\widehat{g}}
\newcommand{\NS}{\mathbb{NS}}
\newcommand{\AS}{\mathbb{AS}}
\renewcommand{\SS}{\mathbb{SS}}
\newcommand{\SO}{\mathsf{SO}}
\newcommand{\frob}[1]{\norm{#1}_{\mathrm{F}}}
\DeclareMathOperator{\vol}{vol}
\newcommand{\bra}[1]{\left\langle #1 \right|}
\newcommand{\ket}[1]{\left| #1 \right\rangle}

\newcommand{\eval}{\mathcal{E}}
\newcommand{\snorm}[1]{\norm{#1}_{S}}
\newcommand{\res}{|_H}

\hypersetup{
  pdfauthor={Cristopher Moore, Alexander Russell}}

\begin{document}

\title{
Heat and Noise on Cubes and Spheres:\\
The Sensitivity of Randomly Rotated 
Polynomial Threshold Functions
}
\author{Cristopher Moore \\ 
Santa Fe Institute \\ \
texttt{moore@santafe.edu}
\and 
Alexander Russell \\ 
University of Connecticut \\ 
\texttt{acr@cse.uconn.edu}}

\maketitle

\abstract{We establish a precise relationship between spherical harmonics and Fourier basis functions over a hypercube randomly embedded in the sphere.  In particular, we give a bound on the expected Boolean noise sensitivity of a randomly rotated function in terms of its ``spherical sensitivity,'' which we define according to its evolution under the spherical heat equation.  As an application, we prove an average case of the Gotsman-Linial conjecture, bounding the sensitivity of polynomial threshold functions subjected to a random rotation.}

\section{Introduction}
\label{sec:intro}

The \emph{average sensitivity} $\AS(f)$ of a Boolean function $f:\Z_2^n \to \{+1,-1\}$ is $n$ times the probability that $f(x) \ne f(y)$, where $x$ is chosen uniformly at random and $y$ is chosen uniformly from $x$'s neighbors at Hamming distance $1$.  Similarly, the \emph{noise sensitivity} $\NS_\eps(f)$ is the probability that $f(x) \ne f(y)$ where $x$ is uniformly random and $y$ is formed by flipping each bit of $x$ independently with probability $\eps$.  Sensitvity is a basic structural characteristic of Boolean functions, with applications to computational complexity, pseudorandomness, machine learning, and the theory of social choice~\cite{LinialMN:1993,KhotKMO:2007,ODonnell:2004,HarshaKM:2014,DiakonikolasSTW:2014,Kalai:2010,odonnell-book}.

Many cases of interest focus on threshold functions $f(x)$ defined from a smooth function $p: \R^n \rightarrow \R$ and a threshold value $\theta \in \R$,
\[
f(x) = \begin{cases}
+1 & \text{if $p(x) \geq \theta$} \, , \\
-1 & \text{if $p(x) < \theta$} \, . 
\end{cases}
\]
where we define a Boolean function by restricting $p$, and $f$, to the hypercube $\{\pm 1\}^n \subset \R$. In particular, polynomial threshold functions $f(x) = \sgn(p(x))$, where $p(x)$ is a polynomial of degree $d$, have played a dominant role in this setting.  
Gotsman and Linial~\cite{gotsman-linial} conjectured that the sensitivity of such functions is maximized when $p$ is a symmetric polynomial whose roots slice the hypercube at $d$ Hamming weights near $n/2$,
\[
p(x_1,\ldots,x_n) = q\!\left( \sum_i x_i \right) 
\quad \text{where} \quad
q(x) = \prod_{h = \lfloor (n-d+1)/2 \rfloor}^{\lfloor (n+d-1)/2 \rfloor} \left( x - h - \frac{1}{2} \right) \, . 
\]
This implies that 
\[
\AS(f) = O(d \sqrt{n}) 
\quad \text{and} \quad
\NS_\eps(f) = O(d \sqrt{\eps}) \, , 
\]
where the constant in the $O$ depends neither on $d$ nor $n$.  This is known in the case $d=1$, i.e., where $f(x)$ is a halfspace~\cite{peres}.  However, for $d > 1$ it has remained open for some time.  

The first nontrivial bounds for threshold functions of degree $d > 1$ were obtained quite recently~\cite{DiakonikolasSTW:2014,HarshaKM:2014}, showing 
\[
\AS(f) = 2^{O(d)} n^{1-\alpha} \log n
\quad \text{and} \quad
\NS_\eps(f) = 2^{O(d)} \eps^\alpha \log(1/\eps) \, , 
\]
where $\alpha = O(1/d)$.  These bounds work by dividing polynomials into two classes: ``juntas'' where a few variables are highly influential, and ``regular'' polynomials where no variable has large influence.  The regular case is handled using anticoncentration bounds and the invariance principle of~\cite{MOO05}, showing that the distribution of values of $p(x)$ is close to what it would be if $x$ were drawn from the Gaussian distribution as opposed to the uniform distribution on the hypercube.  

Using different reasoning~\cite{kane-correct-exponent}, it was recently shown that 
\[
\AS(f) = \sqrt{n} \,(\log n)^{O(d \log d)} \,2^{O(d^2 \log d)}
\quad \text{and} \quad 
\NS_\eps(f) = \sqrt{\eps} \,(\log (1/\eps))^{O(d \log d)} \,2^{O(d^2 \log d)} \, .
\]
While this dependence on $d$ is somewhat regrettable, these results show that the Gotsman-Linial conjecture holds, up to polylogarithmic factors, for each fixed $d$.

Even when our ultimate questions pertain to sensitivity on the hypercube, working with functions defined on $\R^n$ can permit techniques from analysis to be brought to bear on the problem. This has motivated interest in continuous notions of noise sensitivity, most notably \emph{Gaussian sensitivity}, 
which is obtained by placing a Gaussian measure on $\R^n$ and applying Gaussian noise.  In this setting, a simple and elegant argument~\cite{kane-surface-area} shows that the Gaussian analog of $\NS_\eps(f)$ is indeed $O(d \sqrt{\eps})$.

In this article, we introduce a notion of \emph{spherical sensitivity} for functions defined on the unit $n$-sphere.  By analyzing how spherical harmonics are carried to Boolean harmonics, i.e., Fourier basis functions over the hypercube $\Z_2^n$, we give a transfer theorem bounding the Boolean sensitivity in terms of the spherical sensitivity.  Our results hold in expectation, when the function, or equivalently the hypercube, 
is randomly rotated in $\R^n$.  In essence, we show that the distribution of angles induced by Boolean noise on the hypercube can be modeled by the effect of Brownian motion on the sphere, or equivalently diffusion driven by the spherical heat equation.

As an application, by bounding the spherical sensitivity of polynomial threshold functions, we establish the Gotsman-Linial conjecture on average, in the following sense: for any polynomial $p$ of degree $d$, if we apply a random rotation $R$ and then restrict to the hypercube, the expected average sensitivity and noise sensitivity of the resulting Boolean threshold function $f(x) = \sgn(Rp(x))$ are $\Exp_R \AS(f) = O(d \sqrt{n})$ and $\Exp_R \NS_\eps = O(\eps \sqrt{n})$ respectively.



\section{Spherical harmonics and the heat equation}

Here and in later sections we repeat information from classic texts~\cite{szego,vilenkin} and two excellent reviews~\cite{frye-efthimiou,gallier}.  Recall that the Laplace operator on $\R^n$ is defined as
\[
\Delta_{\R^n} 
= \sum_{i=1}^n \frac{\partial^2}{\partial x_i^2} \, . 
\]
The Laplace-Beltrami operator on $S_{n-1}$ consists of the contribution to $\Delta_{\R^n}$ arising from the dependence of a function on angular variables rather than the distance from the origin.  It can be defined by writing $\Delta_{\R^n}$ in polar coordinates, 
\[
\Delta_{\R^n} = \frac{\partial}{\partial r^2} + \frac{n-1}{r} \frac{\partial}{\partial r} + \frac{1}{r^2} \Delta_{S_{n-1}} \, . 
\]
A function $h:\R^n \to \R$ is \emph{harmonic} if $\Delta_{\R^n} h = 0$.  In particular, for each $\ell \ge 0$, there is a linear subspace of harmonic homogeneous polynomials of degree $\ell$ of dimension (where we assume $n \ge 3$)
\begin{equation}
\label{eq:dim}
d_\ell 
= {n + \ell - 1 \choose \ell} - {n + \ell - 3 \choose \ell - 2}
= \frac{n+2\ell-2}{n-2} {n + \ell -3 \choose \ell} \, . 
\end{equation}  
Restricting these polynomials to $S_{n-1}$ gives the so-called \emph{spherical harmonics}.  We denote a basis for these as $\{ Y_{\ell,j} \mid 1 \le j \le d_\ell \}$.  They are eigenfunctions of the Laplace-Beltrami operator: 
\begin{equation}
\label{eq:eigen-spherical}
\Delta_{S_{n-1}} Y_{\ell,j} = -\ell (n+\ell-2) Y_{\ell,j} \, . 
\end{equation}
Any function $g$ in $L_2(S_{n-1})$ can be expanded in terms of spherical harmonics, 
\begin{equation}
\label{eq:hatf-spherical}
g = \sum_{\ell \ge 0} g_\ell = \sum_{\ell \ge 0} \sum_{j=1}^{d_\ell} \hatg(\ell,j) \,Y_{\ell,j} \, . 
\end{equation}
This is analogous to the expansion into Fourier series over $[0,1)
\subset \R$, where small $\ell$ corresponds to smooth, low-frequency
variations, and larger $\ell$ corresponds to higher frequencies.
  
%

\section{Noise sensitivity and the heat equation on the cube and the sphere}

The standard notion of noise sensitivity for a Boolean function $f:\Z_2^n \to \{\pm 1\}$ is as follows~\cite{odonnell-book}.  Define a linear operator $K_\eps$ on the space of probability distributions over $\Z_2^n$ that independently flips each bit with probability $\eps$.  That is, if $d(x,y)$ denotes the Hamming distance, 
\[
K_\eps(x,y) = \eps^{d(x,y)} (1-\eps)^{n-d(x,y)} \, .  
\]
Then if $\delta_x$ is the Kronecker delta function where $\delta_x(z) = 1$ if $z=x$ and $0$ otherwise, 
\[
\NS_\eps(f) = \Exp_{x,y} \Pr[f(y) \ne f(x)] \, ,
\]
where $x$ is uniform in $\{\pm 1\}^n$ and $y$ is chosen from the distribution $K_\eps \delta_x$.  Equivalently, if we define the inner product on the cube as 
\[
\langle f, g \rangle = \Exp_{x \in \Z_2^n} f(x)^* g(x) \, ,
\]
then
\[
\NS_\eps(f) = \frac{1}{2} \left( 1 - \langle f, K_\eps f \rangle \right) \, .
\]

We write $f$ in the Fourier basis, expanding in terms of the characters $\chi_k$, 
\[
f(x) = \sum_{k \in \Z_2^n} \hatf(k) \,\chi_k(x) 
\quad \text{where} \quad
\chi_k(x) = (-1)^{k \cdot x} 
\quad \text{and} \quad
\hatf(k) = \langle f, \chi_k \rangle \, . 
\]
Then we can use the fact that $\chi_k$ is an eigenvector of $K_\eps$, 
\begin{equation}
\label{eq:phi-eigen}
K_\eps \chi_k = (1-2\eps)^{|k|} \chi_k \, ,
\end{equation}
where $|k|$ denotes the Hamming weight of the frequency vector $k$.  Then we obtain
\begin{align}
\NS_\eps(f) 
&= \frac{1}{2} \left( 1 - \sum_{k \in \Z_2^n} \abs{\hatf(k)}^2 (1-2\eps)^{|k|} \right) \nonumber \\
&= \frac{1}{2} \sum_{k \in \Z_2^n} \abs{\hatf(k)}^2 \left( 1-(1-2\eps)^{|k|} \right) \, .
\label{eq:ns}
\end{align}
Here we used the fact that $\langle f, f \rangle = \sum_k \abs{\hatf(k)}^2 = 1$ since $\abs{f(x)}^2=1$ for all $x$.  However, we can also take~\eqref{eq:ns} as the definition of $\NS_\eps(f)$, in which case it can be applied to any function $f:\Z_2^n \to \C$. 

We can also write the noise sensitivity in terms of a continuous-time (but discrete-space) heat equation on the hypercube.  Let $A$ be the adjacency matrix of the hypercube, and let
\[
L = A - n \id
\]
be the graph Laplacian.  If we apply the heat equation
\[
\frac{\partial f}{\partial t} = Lf 
\]
for time $\eps'$ with initial condition $f(0)=f$, we have
\[
f(\eps') 
= \e^{\eps' L} f \, . 
\]
The characters $\chi_k$ are eigenfunctions of the Laplacian, 
\[
L \chi_k = -2|k| \chi_k 
\quad \text{and} \quad
\e^{\eps' L} \chi_k = \e^{-2|k|\eps'} \chi_k \, .
\]
Matching these eigenvalues with~\eqref{eq:phi-eigen}, we see that if
\begin{equation}
\label{eq:eps-prime}
\eps' = \frac{1}{2} \log \frac{1}{1-2\eps} = \eps + O(\eps^2) \, ,
\end{equation}
then 
\[
K_\eps = \e^{\eps' L} \, , 
\]
and
\begin{align}
\NS_\eps 
&= \frac{1}{2} \left( 1-\langle f, f(\eps') \rangle \right) \nonumber \\
&= \frac{1}{2} \left( 1-\langle f, \e^{\eps' L} f \rangle \right) \nonumber \\
&= \frac{1}{2} \left( 1 - \sum_{k \in \Z_2^n} \abs{\hatf(k)}^2 \,\e^{-2 \eps' |k|} \right) \nonumber \\
&= \frac{1}{2} \sum_{k \in \Z_2^n} \abs{\hatf(k)}^2 \left( 1 - \e^{-2 \eps' |k|} \right) \, .
\label{eq:cube-heat}
\end{align}

In analogy with this heat-equation picture of the Boolean noise sensitivity, we define the \emph{spherical sensitivity} $\SS_t(g)$ of a function $g:S_{n-1} \to \{ \pm 1 \}$ as follows.  First define the inner product of two functions $f, g:S_{n-1} \to \C$ as
\begin{equation}
\label{eq:spherical-inner}
\langle f, g \srangle 
= \Exp_{x \in S_{n-1}} f(x)^* \,g(x) 
= \frac{1}{\Omega_{n-1}} \int_{S_{n-1}} \dx \,f(x)^* g(x) \, ,
\end{equation}
where $\Omega_{n-1}$ denotes the surface area of $S_{n-1}$, 
\begin{equation}
\label{eq:omega}
\Omega_{n-1} = \frac{2 \pi^{n/2}}{\Gamma(n/2)} \, . 
\end{equation} 
The heat equation on $S_{n-1}$ is
\[
\frac{\partial g}{\partial t} = \Delta_{S_{n-1}} g \, ,
\]
and applying it for time $t$ with the initial condition $g(0)=g$ gives
\[
g(t) = \e^{\Delta_{S_{n-1}} t} g \, . 
\]
Then we define $\SS_t(g)$ as
\begin{align}
\SS_t(g) 
&= \frac{1}{2} \left( 1-\langle g, g(t) \srangle \right) \nonumber \\
&= \frac{1}{2} \left( 1-\langle g, \e^{t \Delta_{S_{n-1}}} g \srangle \right) \nonumber \\
&= \frac{1}{2} \sum_{\ell,j} \abs{\hatg(\ell,j)}^2 \left( 1-\e^{-t \ell (n+\ell-2)} \right) \, ,
\label{eq:ss}
\end{align}
where we used~\eqref{eq:eigen-spherical} and the expansion~\eqref{eq:hatf-spherical}.  

Here we assumed that $g$ takes values in $\{ \pm 1 \}$.  In that case, $\SS_t(g)$ is the probabiility that $g(x) \ne g(y)$ if $x$ is uniformly random and $y$ is the position of a particle that starts at $x$ and unergoes Brownian motion for time $t$.  However, as with~\eqref{eq:ns}, we will take~\eqref{eq:ss} as the definition of $\SS_t(g)$, thus extending the notion of sensitivity to arbitrary functions $g: S_{n-1} \to \C$.  

Comparing~\eqref{eq:cube-heat} and~\eqref{eq:ss}, we see that flipping bits with probability $\eps$ is roughly analogous to running the heat equation on the sphere for time $t = O(\eps/n)$.  We will tighten this analogy in Theorem~\ref{thm:transfer} below.

\section{Zonal harmonics and Gegenbauer polynomials}
\label{sec:zonal}

In this section and the next, we continue our review of spherical harmonics and their associated orthogonal polynomials~\cite{szego,vilenkin,frye-efthimiou,gallier}.   We include somewhat more machinery than is strictly necessary to prove our main result.  However, in many cases this machinery gives us a more explicit picture of what is going on, and may be useful in proving more detailed results.

Let $\eta = (1,0,\ldots,0)$ be the north pole.  If $f(Rx) = f(x)$ for all rotation matrices that fix $\eta$, then $f(x)$ depends only on $z = \eta \cdot x$.  Such functions are called \emph{zonal}.  The inner product of two such functions can be written as a weighted inner product over the interval $-1 \le z \le 1$, 
\begin{equation}
\label{eq:weighted-inner}
\langle f, g \srangle 
= \frac{\Omega_{n-2}}{\Omega_{n-1}} \int_{-1}^1 \dz \,w_\alpha(z) \,f(z)^* g(z) \, ,
\end{equation}
where we wantonly abuse notation by identifying $f(z)$ with $f(x)$, and the weight 
\begin{equation}
\label{eq:weight}
w_\alpha(z) = (1-z^2)^{\alpha-\frac{1}{2}}
\end{equation}
and the constant $\Omega_{n-2}$ account for the volume of the annulus between height $z$ and $z+\dz$.  

There is a unique harmonic polynomial of each degree $\ell$.  A nice orthogonal family of such polynomials, called the \emph{zonal spherical harmonics} or the \emph{ultraspherical} or \emph{Gegenbauer} polynomials, are as follows.  For historical reasons, we parametrize them with a half-integer $\alpha$ rather than the integer dimension $n$: 
\[
\alpha = \frac{n}{2} - 1 \, ,
\]
in which case the dimension~\eqref{eq:dim} becomes
\begin{equation}
\label{eq:dim-alpha}
d_\ell = \frac{\alpha+\ell}{\alpha} {2\alpha + \ell - 1 \choose \ell} \, . 
\end{equation}
For each value of $\alpha$, we have a family of polynomials $\{ \gamma_\ell \mid \ell \in \N \}$ of degree $\ell$, 
\begin{align}
\label{eq:gegen}
\gamma_\ell(z)
= \frac{1}{\sqrt{N^{(\ell)}}} \,G^{(\ell)}(z) 
\quad \text{where} \quad
G^{(\ell)}(z) 
&= \sum_{k=0}^{\lfloor \ell/2 \rfloor} (-1)^k \frac{(\ell-k+\alpha-1)!}{(\alpha-1)! \,k! \,(\ell-2k)!} (2z)^{\ell-2k} \\
\text{and} \quad
N^{(\ell)} 
& = \left\langle G^{(\ell)}, G^{(\ell)} \right\rangle \nonumber \\
&= \frac{\Omega_{n-2}}{\Omega_{n-1}} \times \frac{\pi \,2^{1-2\alpha} \,\Gamma(\ell+2\alpha)}{(\ell+\alpha) \,\Gamma(\ell+1) \,\Gamma(\alpha)^2} \nonumber \\
&= \frac{\alpha}{\alpha+\ell} 
{2\alpha + \ell - 1 \choose \ell} \, . 
\end{align}
Note that $G^{(\ell)}$ is the $\ell$th Legendre polynomial when $\alpha = 1/2$ (i.e., when $n=3$).  Then for each $\alpha$, the $\gamma_\ell$ are orthonormal with respect to the inner product~\eqref{eq:weighted-inner}:
\[
\left\langle \gamma_\ell, \gamma^{(m)} \right\srangle 
= \delta_{\ell,m} \, . 
\]
Thus, given a zonal function $f(z)$, we can write
\[
f(z) = \sum_{\ell \ge 0} \hatf^{(\ell)} \gamma_\ell(z)
\quad \text{where} \quad
\hatf^{(\ell)} 
= \left\langle f, \gamma_\ell \right\srangle \, . 
\]
This transform is unitary, so inner products are preserved:
\[
\langle f, g \srangle = \sum_{\ell \ge 0} \hatf^{(\ell)*} \,\hatg^{(\ell)} \, .
\]

The Gegenbauer polynomials have deep roots in the representation
theory of Lie groups.  Let $\SO_n$ be the group of orthogonal
rotations of $\R^n$; then the harmonic polynomials of degree $\ell$
form an irreducible representation $\rho_\ell$ of $\SO_n$ with
dimension $d_\ell$.  We can think of $L_2(S_{n-1})$ as the subspace of
$L_2(\SO_n)$ consisting of functions that are right-invariant under
the subgroup $\SO_{n-1}$ that fixes the north pole $\eta$: that is,
functions $f(R)$ that only depend on $R\eta$.  The Gegenbauer
polynomials span the subspace of $L_2(\SO_n)$ consisting of functions
which are left- and right-invariant under $\SO_{n-1}$; that is,
functions that are zonal, depending only on $z = \langle \eta \cdot R\eta \rangle$, or equivalently on the latitude of $R\eta$.  The fact that $\gamma_\ell$ is the unique zonal polynomial of degree $\ell$ corresponds to the fact that $\SO_n$ and $\SO_{n-1}$ form a Gel'fand pair, i.e., this subspace is one-dimensional.

\section{Schur's lemma and evaluation maps}
\label{sec:schur}

Of course, $\eta$ is is an arbitrary choice for the north pole.  
For any $w \in S_{n-1}$, there is a unique polynomial of degree $\ell$ which is zonal around $w$, i.e., which is fixed under the copy of $\SO_{n-1}$ that preserves $w$.  Given a function $f$ and $R \in \SO_n$, define $Rf$ as the function 
\[
Rf(x) = f(R^{-1} x) \, .
\]
Now let $R \in \SO_n$ be any rotation such that $R\eta = w$.  Then if we write $\gamma_\ell^{(\eta)}(x) = \gamma_\ell(\eta \cdot x)$, we can define $\gamma_\ell^{(w)} = R \gamma_\ell^{(\eta)}$, so that
\[
\gamma_\ell^{(w)}(x) = \gamma_\ell^{(\eta)}(R^{-1} x) = \gamma_\ell(w \cdot x) \, .
\]
The inner products of these functions are again given by a Gegenbauer polynomial: for any $w, y \in S_{n-1}$, 
\begin{equation}
\label{eq:gegen-inner}
\left\langle \gamma_\ell^{(w)}, \gamma_\ell^{(y)} \right\srangle
= \frac{1}{\Omega_{n-1}} \int_{S_{n-1}} \dx \,\gamma_\ell(w \cdot x) \,\gamma_\ell(y \cdot x)
= \frac{1}{\sqrt{d_\ell}} \,\gamma_\ell(w \cdot y) \, . 
\end{equation}
Taking $w=y$ tells us how large $\gamma_\ell$ gets at the poles:
\begin{equation}
\label{eq:gegen-max} 
\gamma_\ell(1) 
= \abs{\gamma_\ell(-1)} 
= \sqrt{d_\ell} \, .
\end{equation}
In addition, if we fix $w$ and take the expectation over $y$ of the inner product squared, we get 
\begin{equation}
\Exp_y \babs{ \left\langle \gamma_\ell^{(w)}, \gamma_\ell^{(y)} \right\srangle }^2 
= \frac{1}{d_\ell} \babs{ \left\langle \gamma_\ell, \gamma_\ell \right\srangle }^2 
= \frac{1}{d_\ell} \, . 
\end{equation}
Since averaging over $y$ is the same as averaging over $R$ according to the Haar measure on $\SO_n$, this is equivalent to
\begin{equation}
\label{eq:schur}
\Exp_R \babs{ \left\langle \gamma_\ell, R \gamma_\ell \right\srangle }^2 
= \frac{1}{d_\ell} \, . 
\end{equation}
This holds more generally for any spherical harmonic $f_\ell$ of degree $\ell$, 
\begin{equation}
\label{eq:schur}
\Exp_R \babs{ \left\langle f_\ell, R f_\ell \right\srangle }^2 
= \frac{\snorm{f_\ell}^2}{d_\ell} \, . 
\end{equation}
where $\snorm{f_\ell}^2 = \langle f_\ell, f_\ell \srangle = \Exp_{x \in S_{n-1}} |f_\ell(x)|^2$.  This is a form of Schur's lemma: for any vector $v$ belonging to an irreducible representation $\rho$ of a group $G$, we have $\Exp_g \abs{ \langle v, gv \rangle }^2 = |v|^2/d_\rho$.  

More generally, let $M$ be a linear operator on $L_2(S_{n-1})$.  Conjugating it with a random rotation yields an operator which commutes with all rotations.  By Schur's lemma any such operator is block diagonal, where each block is a scalar matrix operating on the degree-$\ell$ spherical harmonics.  Thus
\[
\Exp_R R^{-1} M R = \bigoplus_{\ell \ge 0} M_\ell 
\quad \text{where} \quad
M_\ell = \frac{\tr M_\ell}{d_\ell} \,\id_\ell \, ,
\]
where $\id_\ell$ is the projection operator onto the space of degree-$\ell$ harmonics.  Thus if $f = \sum_{\ell \ge 0} f_\ell$ where each $f_\ell$ is a spherical harmonic of degree $\ell$, 
\begin{align}
\Exp_R \,\langle Rf, M \cdot Rf \srangle
&= \left\langle f, \left( \Exp_R R^{-1} M R \right) f \right\srangle \nonumber \\
&= \sum_{\ell \ge 0} \langle f_\ell, M_\ell f_\ell \srangle \nonumber \\
&= \sum_{\ell \ge 0} \,\Exp_R \,\langle Rf_\ell, M \cdot Rf_\ell \srangle \, . 
\label{eq:schur-separable}
\end{align}
In particular, if $f = \sum_\ell f_\ell$ then the expected outer product of $Rf$ with itself is
\begin{equation}
\label{eq:schur-outer}
\Exp_R \ket{Rf} \bra{Rf}
= \sum_{\ell \ge 0} \frac{\snorm{f_\ell}^2}{d_\ell} \,\id_\ell \, .
\end{equation}
Thus if $g = \sum_\ell g_\ell$ and $h = \sum_\ell h_\ell$, 
\begin{equation}
\label{eq:schur-outer-inner}
\Exp_R \,\langle g, Rf \srangle \langle Rf, h \srangle 
= \sum_{\ell \ge 0} \frac{\snorm{f}^2}{d_\ell} \langle g_\ell, h_\ell \srangle \, . 
\end{equation}


Another consequence of the irreducibility of $\rho_\ell$ is that any linear operator from $\rho_\ell$ to $\C$ can be written as $\langle \phi, \cdot \srangle$ where $\phi = \sum_{y \in S_{n-1}} c_y \gamma_\ell^{(y)}$ for some finite set of nonzero coefficients $c_y$.  In particular, for any $y \in S_{n-1}$ and any $\ell \ge 0$ there is an \emph{evaluation map} $\eval^{(y)}_\ell: \rho_\ell \to \C$ such that, for any spherical harmonic $f$ of degree $\ell$, we have $\eval^{(y)}_\ell(f) = f(y)$.  We can express it as
\begin{equation}
\label{eq:eval}
\eval^{(y)}_\ell(f) = f(y) = \sqrt{d_\ell} \,\left\langle \gamma_\ell^{(y)} , f \right\srangle \, . 
\end{equation}
To see this, think of $y$ as the north pole, and consider an orthonormal basis that includes $\gamma_\ell^{(y)}$.  Since $\gamma_\ell^{(y)}$ is the unique harmonic of degree $\ell$ that is zonal around $y$, all other basis functions are zero at $y$; otherwise they would have a nonzero projection onto $\gamma_\ell^{(y)}$ if we average over the subgroup $\SO_{n-1}$ of rotations that preserve $y$.  The normalization $\sqrt{d_\ell}$ then follows from~\eqref{eq:gegen-max}.

By summing over all $\ell$, we can similarly express the evaluation map for all integrable functions $f \in L_2(S_{n-1})$.  That is, we can define the evaluation map
\[
\eval^{(y)} = \sum_{\ell \ge 0} \eval^{(y)}_\ell \, . 
\]
Then
\begin{equation}
\label{eq:eval-all}
\eval^{(y)}(f)
= f(y) 
= \sum_{\ell \ge 0} \sqrt{d_\ell} \,\left\langle \gamma_\ell^{(y)} , f \right\srangle \, . 
\end{equation}
To put it differently, we can express the Dirac delta function as
\begin{equation}
\label{eq:dirac}
\delta(x-y) 
= \sum_{\ell \ge 0} \sqrt{d_\ell} \,\gamma_\ell^{(y)}(x) 
= \sum_{\ell \ge 0} \sqrt{d_\ell} \,\gamma_\ell(y \cdot x) \, . 
\end{equation}



\section{Relating noise sensitivity and spherical sensitivity for randomly rotated functions}

In this section we will prove our main transfer theorem, bounding the expected noise sensitivity of a randomly rotated function in terms of its spherical sensitivity.

We identify the hypercube $\Z_2^n$ with the set $H = \{\pm 1/\sqrt{n} \}^n$ lying on the unit sphere.  If we restrict the inner product to this set, we obtain a cubical inner product, which we write 
\[
\langle f, g \crangle 
= \Exp_{x \in H} f(x)^* g(x) \, .
\]
In particular, for each frequency vector $k \in \Z_2^n$, if we extend the character $\chi_k$ to the sphere as a multilinear function of degree $|k|$, 
\[
\chi_k(x) = \prod_{i: k_i=1} \sqrt{n} x_i \, , 
\]
then the $\chi_k$ are orthonormal with respect to the cubical inner product, 
\[
\langle \chi_k, \chi_{k'} \crangle = \delta_{k,k'} \, . 
\]
If we define a Boolean function $f\res$ by restricting a function $f$ to the hypercube, its Fourier coefficients are
\[
\hatf\res(k) = \langle f, \chi_k \crangle \, .
\]
The \emph{energy} of $f$ at the character $k$ is $|\hatf(k)|^2$.  

In order to bound the noise sensitivity, we need to compute the expected energy of $Rf$ where $R$ is uniformly random, i.e., chosen according to the Haar measure in $\SO_n$; or equivalently, the expected energy of $f$'s restriction to a randomly rotated hypercube.  One simple observation is the following.  Since uniformly rotating any point on the cube yields a uniformly random point on the sphere, cubical inner products are equal to spherical inner products in expectation,
\begin{equation}
\label{eq:cube-sphere-exp}
\Exp_R \langle Rf, Rg \crangle = \langle f, g \srangle \, . 
\end{equation}

In Appendix~\ref{app:kravchuk} we give a precise expression for the expected energy of a randomly rotated spherical harmonic.  However, here we just need a few facts.  First,  if we decompose a function into spherical harmonics, then its expected energy at each $k$ is the sum of the expected energies of its harmonics:
\begin{lemma}
\label{lem:separable}
Let $f \in L_2(S_{n-1})$, and write $f = \sum_{\ell \ge 0} f_\ell$ where each $f_\ell$ is a spherical harmonic of degree $\ell$.  Let $R \in \SO_n$ be uniform in the Haar measure.  Then for any $k \in \Z_2^n$, 
\[
\Exp_R \babs{ \widehat{Rf}\res(k) }^2 
= \sum_{\ell \ge 0} \,\Exp_R \,\babs{ \widehat{Rf_\ell}\res(k) }^2 \, . 
\]
\end{lemma}

\begin{proof}
First we use the evaluation maps of Section~\ref{sec:schur} to write $\widehat{f}(k)$ as a spherical inner product rather than a cubical one.  Using~\eqref{eq:eval-all} we have
\[
\hatf(k)
= \langle \chi_k, f \crangle 
= \Exp_{x \in H} \chi_k(x) \,\eval^{(x)}(f)
= \left\langle \psi_k, f \right\srangle \, , 
\]
where
\[
\psi_k = \Exp_{x \in H} \chi_k(x) \sum_{\ell \ge 0} \sqrt{d_\ell} \gamma_\ell^{(x)} \, . 
\]
Then applying Schur's lemma~\eqref{eq:schur-separable} to the linear operator $M = \ket{\psi_k} \bra{\psi_k}$ gives
\begin{align*}
\Exp_R \babs{ \widehat{Rf}\res(k) }^2 
&= \Exp_R \,\langle Rf, \psi_k \srangle \langle \psi_k, Rf \srangle \\
&= \sum_{\ell \ge 0} \,\Exp_R \,\langle Rf_\ell, \psi_k \srangle \langle \psi_k, Rf_\ell \srangle \\
&= \sum_{\ell \ge 0} \,\Exp_R \,\babs{ \widehat{Rf_\ell}\res(k) }^2 \, ,
\end{align*}
completing the proof.
\end{proof}

Secondly, restricting a spherical harmonic of degree $\ell$ to the hypercube can only give it nonzero energy at characters of Hamming weight $\ell$ or less:

\begin{lemma}
\label{lem:lower}
Let $f$ be a spherical harmonic of degree $\ell$.  If $|k| \ge \ell$, then $\hatf\res(k) = 0$. 
\end{lemma}

\begin{proof}
Restricting any polynomial $f$ of degree $\ell$ to $H = \{ \pm 1/\sqrt{n} \}^n$ imposes the relations $x_i^2 = 1/n$ for all $i$, so there is a multilinear polynomial $f'$ of degree at most $\ell$ such that $f\res = f'\res$.  Each multilinear monomial is proportional to a character $\chi_k$ with $|k| \le \ell$, and these are orthogonal to all $\chi_k$ with $|k| > \ell$.
\end{proof}

Thirdly, for any integrable function $f$, the expected energy of $Rf$ summed over all characters of $\Z_2^n$ equals its norm on the sphere:
\begin{lemma}
\label{lem:total}
Let $f \in L_2(S_{n-1})$, and let $R \in \SO_n$ be uniform in the Haar measure.  Then 
\[
\Exp_R \sum_{k \in \Z_2^n} \babs{ \widehat{Rf}\res(k) }^2 = \snorm{f}^2 \, .
\]
\end{lemma}

\begin{proof}
Since the $\chi_k$ are orthonormal with respect to the inner product on the hypercube, for any $f$ we have
\[
\sum_{k \in \Z_2^n} \babs{ \hatf\res(k) }^2 
= \sum_k \langle f, \chi_k \crangle \langle \chi_k, f \crangle 
= \langle f, f \crangle 
\]
Then~\eqref{eq:cube-sphere-exp} gives
\[
\Exp_R \sum_{k \in \Z_2^n} \babs{ \widehat{Rf}\res(k) }^2 
= \Exp_R \langle Rf, Rf \crangle 
= \langle f, f \srangle 
= \snorm{f}^2 \, .\qedhere 
\]
\end{proof}

We are now ready to prove our transfer theorem, which bounds the expected noise sensitivity and average sensitivity in terms of the spherical sensitivity.
\begin{theorem}
\label{thm:transfer}
Let $f \in L_2(S_{n-1})$, and let $R \in \SO_n$ be uniform in the Haar measure.  Then
\begin{equation}
\label{eq:transfer-ns}
\Exp_R \NS_\eps(Rf\res) 
\le \SS_t(f) \, ,
\end{equation}
where
\begin{equation}
\label{eq:t}
t 
= \frac{1}{n} \log \frac{1}{1-2\eps} 
= \frac{2\eps}{n} \big(1+O(\eps)\big) \, . 
\end{equation}
\end{theorem}

\begin{proof}
Starting with the Fourier-theoretic expression for the noise sensitivity~\eqref{eq:ns}, we have
\begin{align}
\Exp_R \NS_\eps(Rf\res)
&= \frac{1}{2} \Exp_R \,\sum_{k \in \Z_2^n} \babs{\widehat{Rf}\res(\chi)}^2 \left(1 - (1 - 2\eps)^{|k|}\right) 
\nonumber \\
&= \frac{1}{2} \Exp_R \,\sum_k \sum_{\ell \ge 0} \babs{\widehat{Rf_\ell}\res(k)}^2 \left(1 - (1 - 2\eps)^{|k|} \right) 
\label{eq:separable} \\
&= \frac{1}{2} \Exp_R \,\sum_\ell \sum_{k: |k| \le \ell} \babs{\widehat{Rf_\ell}\res(k)}^2 \left(1 - (1 - 2\eps)^{|k|} \right) 
\label{eq:lower} \\
&\le \frac{1}{2} \Exp_R \,\sum_\ell \sum_k \babs{\widehat{Rf_\ell}\res(k)}^2 \left(1 - (1 - 2\eps)^{\ell} \right) 
\nonumber \\
&\le \frac{1}{2} \sum_\ell \snorm{f_\ell}^2 \left(1 - (1 - 2\eps)^{\ell} \right) \, .
\label{eq:total}
\end{align}
Here we used Lemma~\ref{lem:separable} in~\eqref{eq:separable}, Lemma~\ref{lem:lower} in~\eqref{eq:lower}, and Lemma~\ref{lem:total} in~\eqref{eq:total}.

On the other hand, using $\snorm{f_\ell}^2 = \sum_j \abs{ \hatf(\ell,j) }^2$ in~\eqref{eq:ss} gives
\begin{align}
\SS_t(f) 
&= \frac{1}{2} \sum_{\ell} \snorm{f_\ell}^2 \left( 1-\e^{-t \ell (n+\ell-2)} \right) 
\nonumber \\
&\ge \frac{1}{2} \sum_{\ell} \snorm{f_\ell}^2 \left( 1-\e^{-t \ell n} \right) \, . 
\label{eq:ss-lower}
\end{align}
Setting $1-2\eps = \e^{-tn}$ completes the proof of~\eqref{eq:transfer-ns}.
\end{proof}

Although we don't need it below, we record an analogous theorem regarding the expected average sensitivity.

\begin{theorem}
\label{thm:transfer-as}
Let $f \in L_2(S_{n-1})$, and let $R \in \SO_n$ be uniform in the Haar measure.  Then for any $\alpha > 0$, 
\begin{equation}
\label{eq:transfer-as}
\Exp_R \AS(Rf\res) 
\le \frac{2n}{1-\e^{-\alpha}} \,\SS_{\alpha/n^2}(f) \, . 
\end{equation}
\end{theorem}

\begin{proof}
Recall the Fourier-theoretic expression for the average sensitivity, 
\begin{align}
\AS(f) 
&= - \frac{1}{2} \langle f, Lf \rangle \nonumber \\
&= \sum_{k \in \Z_2^n} \abs{\hatf(k)}^2 |k| \, . 
\label{eq:as}
\end{align}

Applying the same lemmas to~\eqref{eq:as} as we did to~\eqref{eq:ns} in the proof of the previous lemma, and noting that $|k| \le n$, gives
\[
\Exp_R \AS(Rf\res)
\le \sum_\ell \snorm{f_\ell}^2 \max(\ell,n) \, .
\]
Setting $t=\alpha/n^2$ in~\eqref{eq:ss-lower}, we have for all $0 \le \ell \le n$
\[
\frac{1-\e^{-\alpha \ell/n}}{1-\e^{-\alpha}} \ge \frac{\ell}{n} \, ,
\]
since $1-\e^{-x}$ is concave.  This completes the proof of~\eqref{eq:transfer-as}.
\end{proof}

\section{Application: the Gotsman-Linial conjecture on average}

In this section we bound the spherical sensitivity of polynomial threshold functions, and apply Theorem~\ref{thm:transfer} to bound their expected noise sensitivity.  

Our strategy for bounding the spherical sensitivity is similar to that of Kane~\cite{kane-surface-area}, who proved a similar bound on their Gaussian sensitivity.  He used the fact that adding Gaussian noise to a point $u$ can be thought of as choosing a random line through $u$, and then moving along that line by a distance $r$ chosen from the chi-squared distribution: then $p(x)$ can only change sign if it has a root on the intervening line segment.  Similarly, we add noise on the sphere by choosing a random great circle that passes through $u$, and then moving an angle $r$ along that circle where $r$ is chosen from a distribution derived from the heat equation on $S_{n-1}$.  The sensitivity is then at most the probability that the polynomial has a root on the resulting segment of the great circle.

First we need to bound the distribution of the angle, or equivalently the geodesic distance, that we travel on the sphere in time $t$.  Given two points $u, v \in S_{n-1}$ and a time $t \ge 0$, let $K_t(u,v)$ denote the \emph{heat kernel} on $S_{n-1}$ after time $t$.  That is, for any fixed $t$ and $u$, $K_t(u,v)$ is the probability distribution of the position $v$ of a particle that starts at $u$ and undergoes Brownian motion for time $t$.  As a linear operator, it is the solution to the partial differential equation
\[
\frac{\partial K_t}{\partial t} = \Delta_{S_{n-1}} K_t 
\]
with the initial condition $K_0 = \id$, so
\[
K_t(u,v) = \e^{t\Delta_{S_{n-1}}} \, . 
\]
By comparing the heat kernel on $S_{n-1}$ to that on $\R^{n-1}$, we prove the following:

\begin{lemma}
\label{lem:sphere-flat}
Fix $u \in S_{n-1}$, and suppose that $v \in S_{n-1}$ is chosen with probability distribution $K_t(u,v)$.  Let $r$ denote the angle between $u$ and $v$.  Then
\[
\Exp[r] \le \sqrt{2(n-1)t} \, . 
\]
\end{lemma}

\begin{proof}
In Appendix~\ref{app:sphere-flat} we offer an elementary calculus proof that $\Exp[r^2] \le 2(n-1)t$, which implies the lemma.  However, we can prove something much stronger: namely, that $r$ is stochastically dominated by the corresponding process on the flat tangent space $\R^{n-1}$.  The heat equation on $\R^{n-1}$ is driven by the Laplacian
\begin{equation}
\label{eq:heat-flat}
\Delta_{\R^{n-1}} 
= \sum_{i=1}^{n-1} \frac{\partial^2}{\partial x_i^2} \, . 
\end{equation}
Place $u$ at the origin, and let $r$ denote the distance from the origin.  Since $K_t(u,v) = f(r)$ is spherically symmetric, transforming to polar coordinates gives
\[
\frac{\partial f}{\partial t} 
= \Delta_{\R^{n-1}} f 
= \frac{\partial^2 \!f}{\partial r^2} + (n-2) \frac{1}{r} \frac{\partial f}{\partial r} \, .
\]
Similarly, for $S_{n-1}$, place $u$ at the north pole and let $r$ denote the angle between $u$ and $v$.  Then $K_t(u,v) = f(r)$ is a zonally symmetric function, and applying the Laplace-Beltrami operator gives
\begin{equation}
\label{eq:heat-sphere}
\frac{\partial f}{\partial t}
= \Delta_{S_{n-1}} f 
= \frac{\partial^2 \!f}{\partial r^2} + (n-2) \frac{\cos r}{\sin r} \frac{\partial f}{\partial r} \, .
\end{equation}

We can view~\eqref{eq:heat-flat} and~\eqref{eq:heat-sphere} as governing the probability distributions of two stochastic processes on $r$.  These are well known in the theory of Brownian motion, and are referred to as Bessel and Jacobi processes respectively.  Since $\cos r / \sin r \le 1/r$, the comparison theorem of stochastic differential equations~\cite{ikeda-watanabe} implies that the distribution of $r$ on $S_{n-1}$ is stochastically dominated by its distribution on $\R^{n-1}$.  In particular, its second moment is at most the variance of $n-1$ independent variables $x_1,\ldots,x_{n-1} \in \R$ of variance $2t$, giving 
\[
\Exp[r^2] \le 2(n-1)t \, . 
\] 
Noting that $\Exp[r] \le \sqrt{\Exp[r^2]}$ completes the proof.
\end{proof}

We remark that for small $n$, we get a small improvement by computing $\Exp[r]$ exactly on $\R^{n-1}$.  The fact that $r^2/2$ follows a chi-squared distribution with $n-1$ degrees of freedom implies
\[
\Exp[r] 
\le \frac{2 \Gamma(n/2)}{\Gamma((n-1)/2)} \sqrt{t}
= \big( 1-O(1/n) \big) \sqrt{2(n-1)t} \, . 
\]

\begin{lemma}
\label{lem:great-circle}
Let $p \in \R[x_1, \ldots, x_n]$ be a polynomial of degree $d$, and let $G$ a great circle on $S_{n-1}$.  If $p$ is not identically zero on $G$, then $p$ has no more than $2d$ roots on $G$.  
\end{lemma}

\begin{proof}
Since applying a linear transformation to $x_1, \ldots, x_n$ doesn't change $p$'s degree, without loss of generality we can assume that $G$ is the unit circle in the plane spanned by the $x_1$ and $x_2$ axes: that is, the variety, or set of roots, of the polynomial 
\[
q(x_1,x_2) = x_1^2 + x_2^2 - 1 \, . 
\]
Restricting to this plane, i.e., setting $x_i=0$ for all $i > 2$, yields a polynomial $r(x_1,x_2)$ of degree $d_r \le d$.  
B\'{e}zout's theorem~\cite{Shafarevich:1994,Schmid:1995} states that two polynomials $q, r$ of degree $d_q$ and $d_r$ can share at most $d_q d_r$ roots unless they share a common factor.  It is easy to check that $q$ is irreducible over $\R$ (and even over $\C$), since if it had a linear factor then $G$ would consist of the union of two lines.  Therefore, $q$ and $r$ share a common factor only if $q$ divides $r$, in which case $p$ is identically zero on $G$.  If they do not, they share at most $2d$ roots.
\end{proof}

Putting these lemmas together gives us a bound on the spherical sensitivity of a polynomial threshold function.
\begin{theorem}
\label{thm:poly-ss}
Let $p: \R^n \to \R$ be a polynomial of degree $d$ and let $f(x) = \sgn(p(x))$ where $\sgn z = +1$ for $z \ge 0$ and $\sgn z = -1$ for $z < 0$.  Then 
\[
\SS_t(f) \le \frac{d}{\pi} \sqrt{2nt}  \, . 
\]
\end{theorem}

\begin{proof}
Recall that $\SS_t(f)$ is the probability that $f(u) \ne f(v)$ where $u$ is chosen uniformly and $v$ is chosen from $K_t(u,v)$.  Equivalently, we can choosen $u$ uniformly, and then arrive at $v$ by choosing a uniformly random great circle $G$ passing through $u$ (by choosing a tangent vector from the Haar measure on $S_{n-2}$), choosing $r$ according to the heat kernel, and moving an angle $r$ along $G$.

If $p$ is identically zero on $S_{n-1}$, then $f$ is identically zero as well, in which case $\SS_t(f) = 0$.  Otherwise, with probability $1$ we have $f(u) \ne 0$, in which case $p$ is not identically zero on $G$.  By Lemma~\ref{lem:great-circle}, there are at most $2d$ roots of $p$ on $G$.  The probability that $p(u)$ and $p(v)$ have different signs is then at most the expected number of roots of $p$ on $G$ between $u$ and $v$.  Since $u$'s position on $G$ is uniformly random, this is simply $(2d/2\pi) \Exp[r]$, which by Lemma~\ref{lem:sphere-flat} is at most $(d/\pi) \sqrt{2nt}$.  
\end{proof}

Theorems~\ref{thm:transfer} and~\ref{thm:poly-ss} immediately imply our bound on the expected noise sensitivity:

\begin{theorem}
\label{thm:poly-ns}
Let $p$ be a polynomial of degree $d$ and let $f(x) = \sgn(p(x))$.  Let $R \in \SO_n$ be uniform in the Haar measure.  Then 
\begin{equation}
\label{eq:poly-ns}
\Exp_R \NS_\eps(Rf\res) \le \big( 1+O(\eps) \big) \frac{2}{\pi} \,d \sqrt{\eps} \, .
\end{equation}
\end{theorem}

Similar but simpler reasoning implies a bound on the expected average sensitivity:

\begin{theorem}
\label{thm:poly-as}
Let $p$ be a polynomial of degree $d$ and let $f(x) = \sgn(p(x))$.  Let $R \in \SO_n$ be uniform in the Haar measure.  Then 
\begin{equation}
\label{eq:poly-as}
\Exp_R \AS(Rf\res) 
\le \big( 1+O(1/n) \big) \frac{2}{\pi} \, d \sqrt{n} \, . 
\end{equation}
\end{theorem}

\begin{proof}
The angle between two adjacent corners of the hypercube is 
\[
r = \cos^{-1} (1-2/n) = \frac{2}{\sqrt{n}} \,\big(1+O(1/n) \big) \, ,
\]
and $\Exp_R \AS(Rf\res)$ is $n$ times the probability that $p$ changes sign between a uniformly random pair of points $r$ apart.  Using Lemma~\ref{lem:great-circle} as in the proof of Theorem~\ref{thm:poly-ns}, this probability is at most $(d/\pi)r$.
\end{proof}

\begin{remark}
For small $\eps$, we can also prove Theorem~\ref{thm:poly-ns} by noting that the angle $r$ between two points $x, y$ on the hypercube where we have flipped each bit independently with probability $\eps$ obeys $\Exp[\cos r] = 1-2\eps$.
\end{remark}

\begin{remark}
The leading constant in~\eqref{eq:poly-as} is better than we would obtain from Theorem~\ref{thm:transfer-as}, even after optimizing the parameter $\alpha$.
\end{remark}

\begin{remark}
Theorems~\ref{thm:poly-ss} and~\ref{thm:poly-ns} immediately generalize from polynomials to any class of functions with a bound on the number of roots lying on a great circle.  If there are at most $b$ such roots, the expected noise sensitivity of the corresponding threshold function is $O(b \sqrt{\eps})$ and its expected average sensitivity is $O(b \sqrt{n})$.
\end{remark}

\paragraph{Acknowledgments}  We are grateful to Fabrice Baudoin, Laura De Carli, Costas Efthimiou, Veit Elser, Josh Grochow, Ilia Krasikov, Ryan O'Donnell, Thomas H.\ Parker, Dan Rockmore, and James Stokes for helpful conversations.  This work was supported by NSF grants CCF-1117426 and CCF-1219117.










\bibliography{sensitivity}

\appendix

\section{The expected energy in terms of the Gegenbauer and Kravchuk polynomials}
\label{app:kravchuk}

For $0 \le k,h \le n$, l et $\kappa_k(h)$ denote the Kravchuk polynomial
\[
\kappa_k(h)
= \sum_{j=0}^k (-1)^h {k \choose j} {n-k \choose h-j} \, . 
\]

\begin{lemma}
\label{eq:exp-harmonic}
Let $f: S_{n-1} \to \C$ be a spherical harmonic of degree $\ell$, and let $R \in \SO_n$ be uniform in the Haar measure.  Then 
\[
\Exp_R \babs{ \widehat{Rf}\res(k) }^2 
= \frac{\snorm{f}^2}{\sqrt{d_\ell}} \frac{1}{2^n} \sum_{h=0}^n \kappa_k(h) \,\gamma_\ell(1-2h/n) \, , 
\]
\end{lemma}

\begin{proof}
Given $x, y \in \Z_2^n$, we define $x \oplus y$ as their sum in $\Z_2^n$.  In terms of the embedding $\{ \pm 1/\sqrt{n} \}$ of $\Z_2^n$ in the unit sphere, $(x \oplus y)_i = \sqrt{n} x_i y_i$.  In that case we have $\chi_k(x \oplus y) = \chi_k(x) \chi_k(y)$.  Then
\begin{align*}
\Exp_R \babs{ \widehat{Rf}\res(k) }^2 
&= \Exp_R \babs{\langle Rf, \chi_k \crangle}^2 \\
 &= \Exp_R \babs{\Exp_{x \in \Z_2^n} \chi_k(x) Rf(x)}^2 \\
 &= \Exp_R \Exp_{x,y \in \Z_2^n} \chi(x)\chi_k(y) Rf(x) Rf(y)^* \\
&= \Exp_{x,y \in \Z_2^n} \chi_k(x \oplus y) \Exp_R Rf(x) Rf(y)^* \, .
\end{align*}
Let us place the north pole $\eta$ at the corner $0^n$ of the hypercube, i.e., at $(1,\ldots,1)/\sqrt{n}$.  Since the angle between $x$ and $y$ is the same as between $\eta$ and $x \oplus y$, we can change variables to $z = x \oplus y$ and write this expectation as
\begin{equation}
\label{eq:random-coefficient}
\Exp_R \babs{ \widehat{Rf}\res(k) }^2 
= \Exp_{z \in \Z_2^n} \chi_k(z) \Exp_R Rf(z) Rf(\eta)^* \, .
\end{equation}

As in Section~\ref{sec:schur}, let $\eval^{(z)}_\ell$ denote the linear operator that evaluates a spherical harmonic of degree $\ell$ at the point $z$.  Writing~\eqref{eq:random-coefficient} in terms of $\eval^{(z)}_\ell$ and using~\eqref{eq:eval}, we have
\begin{align*}
\Exp_R \babs{ \widehat{Rf}\res(k) }^2 
&= \Exp_{z \in \Z_2^n} \chi_k(z) \Exp_R \eval^{(z)}_\ell(Rf) \,\eval^{(\eta)}_\ell(Rf)^* \\
&= d_\ell \Exp_{z \in \Z_2^n} \chi_k(z) 
  \Exp_R \left[ \left\langle \gamma_\ell^{(z)}, Rf \right\srangle \left\langle Rf, \gamma_\ell^{(\eta)} \right\srangle \right] \, ,
\end{align*}
Now we apply Schur's lemma~\eqref{eq:schur-outer-inner} and the formula~\eqref{eq:gegen-inner} for the inner product of two rotated Gegenbauer polynomials to obtain
\begin{align}
\Exp_R \babs{ \widehat{Rf}\res(k) }^2 
&= \snorm{f}^2 \Exp_{z \in \Z_2^n} \chi_k(z) \left\langle \gamma_\ell^{(z)}, \gamma_\ell^{(\eta)} \right\srangle \nonumber \\
&= \frac{\snorm{f}^2}{\sqrt{d_\ell}} \Exp_{z \in \Z_2^n} \chi_k(z) \,\gamma_\ell^{(\eta)}(z) \, . 
\label{eq:random-coefficient-2}
\end{align}
Finally, we write the expectation over $z$ as a sum over $z$'s Hamming weight $h$.  For each $0 \le j \le |k|$, there are ${k \choose j} {n-k \choose h-j}$ points $z \in \Z_2^n$ with Hamming weight $h$ and with $j$ ones in the support of $k$, in which case $\chi_k(k) = (-1)^j$.  Thus 
\[
\sum_{\substack{z \in \Z_2^n \\ |z|=h}} \chi_k(z) 
= \sum_{j=0}^k (-1)^j {k \choose j} {n-k \choose h-j} 
= \kappa_k(h) \, . 
\]
Finally, we have $\eta \cdot z = 1 - 2h/n$.  Thus
\[
\Exp_R \babs{ \widehat{Rf}\res(k) }^2 
= \frac{\snorm{f}^2}{\sqrt{d_\ell}} \frac{1}{2^n} \sum_{h=0}^n \kappa_k(h) \,\gamma_\ell(1-2h/n) \, ,
\]
completing the proof.
\end{proof}

\section{A calculus proof of Lemma~\ref{lem:sphere-flat}}
\label{app:sphere-flat}

\begin{proof}
As in the main text, applying the Laplacian on $\R^{n-1}$ to a spherically-symmetric function $K_t(u,v)=f(r)$ gives
\[
\Delta_{\R^{n-1}} f = \frac{\partial^2 \!f}{\partial r^2} + (n-2) \frac{1}{r} \frac{\partial f}{\partial r} \, .
\]
If $f(u,v)=f(r)$ is a probability distribution on $v$, the second moment of $r$ is an integral over spherical shells of radius $r$.  A shell of thickness $\dr$ has volume $\Omega_{n-2} r^{n-2} \,\dr$, giving 
\begin{gather*}
1 = \Omega_{n-2} \int_0^\infty f(r) \,r^{n-2} \,\dr \\
\Exp[r^2] = \Omega_{n-2} \int_0^\infty f(r) \,r^n \,\dr \, .
\end{gather*}
If $f$ obeys the differential equation
\[
\frac{\partial f}{\partial t} = \Delta_{R^{n-1}} f \, ,
\]
then integrating by parts gives
\begin{align*}
\frac{\partial}{\partial t} \Exp[r^2] 
&= \Omega_{n-2} \int_0^\infty \Delta_{\R^{n-1}} f(r) \,r^n \,\dr \\
&= \Omega_{n-2} \left[ 
\int_0^\infty \frac{\partial^2 \!f}{\partial r^2} \,r^n \,\dr 
+ (n-2) \int_0^\infty \frac{\partial f}{\partial r} \,r^{n-1} \,\dr 
\right] \\
&= \Omega_{n-2} \left[ 
n(n-1) 
\int_0^\infty f(r) \,r^{n-2} \,\dr 
- (n-2)(n-1) \int_0^\infty f(r) \,r^{n-2} \,\dr 
\right] \\
&= 2(n-1) \,\Omega_{n-2} \int_0^\infty f(r) \,r^{n-2} \,\dr \\
&= 2(n-1) \, .
\end{align*}
Integrating over $t$ recovers the fact that, after undergoing Brownian motion for time $t$, 
\[
\Exp[r^2] = \sum_{i=1}^{n-1} \Exp[x_i^2] = 2(n-1)t \, . 
\]

Now we carry out the same calculation for $S_{n-1}$.  Place $u$ at the north pole, and let $r$ denote the geodesic distance (equivalently, the angle) between $u$ and $v$.  Then $K_t(u,v) = f(r)$ is a zonal function, and applying the Laplace-Beltrami operator gives
\[
\Delta_{S_{n-1}} f = \frac{\partial^2 \!f}{\partial r^2} + (n-2) \frac{\cos r}{\sin r} \frac{\partial f}{\partial r} \, . 
\]
We again integrate over annular shells at distance $r$.  A shell of angular thickness $\dr$ has volume $\Omega_{n-2} \sin^{n-2} r \,\dr$, so 
\begin{gather*}
1 = \Omega_{n-2} \int_0^\pi f(r) \,\sin^{n-2} r \,\dr \\
\Exp[r^2] = \Omega_{n-2} \int_0^\pi f(r) \,r^2 \sin^{n-2} r \,\dr \, .
\end{gather*}
Applying 
\[
\frac{\partial f}{\partial t} = \Delta_{S_{n-1}} f \, ,
\]
a lengthier integration by parts (keeping in mind that the boundary terms are zero) yields
\begin{align*}
\frac{\partial}{\partial t} \Exp[r^2] 
&= \Omega_{n-2} \int_0^\pi \Delta_{\R^{n-1}} f(r) \,r^2 \sin^{n-2} r \,\dr \\
&= \Omega_{n-2} \left[ 
\int_0^\pi \frac{\partial^2 \!f}{\partial r^2} \,r^2 \sin^{n-2} r \,\dr 
+ (n-2) \int_0^\pi \frac{\partial f}{\partial r} \,r^2 \sin^{n-3} r \cos r \,\dr 
\right] \\
&= 2 \,\Omega_{n-2} \int_0^\pi f(r) (\sin^{n-3} r) (\sin r + (n-2) r \cos r) \,\dr \\
&\le 2 (n-1) \,\Omega_{n-2} \int_0^\pi f(r) \,\sin^{n-2} r \,\dr \\
&= 2(n-1) \, . 
\end{align*}
where in the second-to-last line we used $r \cos r \le \sin r$ for $r \in [0,\pi]$.  Integrating over $t$ then gives 
\[
\Exp[r^2] \le 2(n-1)t \, ,
\]
and noting that $\Exp[r] \le \sqrt{\Exp[r^2]}$ completes the proof.
\end{proof}

\end{document}